\documentclass[a4paper]{article}

\usepackage{a4wide}
\usepackage[english]{babel}
\usepackage{graphicx}
\usepackage{amsmath}
\usepackage{amssymb}
\usepackage{xspace}
\usepackage{microtype}
\usepackage{caption}
\usepackage{xcolor}
\usepackage{array}
\usepackage{longtable}

\graphicspath{{./images/}}

\newtheorem{defin}{Definition}
 
\newtheorem{theo}[defin]{Theorem}
 \newenvironment{theorem}{\begin{theo} \sl}{\end{theo}}
\newtheorem{lem}[defin]{Lemma}
 \newenvironment{lemma}{\begin{lem} \sl}{\end{lem}}
\newtheorem{coro}[defin]{Corollary}
 \newenvironment{corollary}{\begin{coro} \sl}{\end{coro}}

\newenvironment{proof}{\emph{Proof.}}{\hfill $\Box$\\}

\newcommand{\etal}{\emph{et~al.}\xspace}

\newcommand{\Graph}[1]{\ensuremath{\theta_{(4 k + #1)}}-Graph\xspace}
\newcommand{\graph}[1]{\ensuremath{\theta_{(4 k + #1)}}-graph\xspace}
\newcommand{\canon}[2]{\ensuremath{T_{#1 #2}}}

\newcommand{\ygraph}{constrained $Y_m$-graph\xspace}

\newcommand{\Vis}{\mathord{\it Vis}}

\newcommand{\const}{\ensuremath{\boldsymbol{c}}\xspace}

\begin{document}

\title{Spanning Properties of Yao and $\Theta$-Graphs in the Presence of Constraints\thanks{Research supported in part by NSERC and Carleton University's President's 2010 Doctoral Fellowship.}~\thanks{Extended abstracts containing results in this paper appeared in LATIN 2014 and CCCG 2014.}}

\author{Prosenjit Bose\thanks{Carleton University, Ottawa, Canada. {\tt jit@scs.carleton.ca}} 
\and Andr\'e van Renssen\thanks{University of Sydney, Sydney, Australia. {\tt andre.vanrenssen@sydney.edu.au}}  }

\date{}

\maketitle

\begin{abstract}
  We present improved upper bounds on the spanning ratio of constrained $\theta$-graphs with at least 6 cones and constrained Yao-graphs with 5 or at least 7 cones. Given a set of points in the plane, a Yao-graph partitions the plane around each vertex into $m$ disjoint cones, each having aperture $\theta = 2 \pi/m$, and adds an edge to the closest vertex in each cone. Constrained Yao-graphs have the additional property that no edge properly intersects any of the given line segment constraints. Constrained $\theta$-graphs are similar to constrained Yao-graphs, but use a different method to determine the closest vertex. 
  
  We present tight bounds on the spanning ratio of a large family of constrained $\theta$-graphs. We show that constrained $\theta$-graphs with $4k + 2$ ($k \geq 1$ and integer) cones have a tight spanning ratio of $1 + 2 \sin(\theta/2)$, where $\theta$ is $2 \pi / (4k + 2)$. We also present improved upper bounds on the spanning ratio of the other families of constrained $\theta$-graphs. These bounds match the current upper bounds in the unconstrained setting. 

  We also show that constrained Yao-graphs with an even number of cones ($m \geq 8$)  have spanning ratio at most $1 / \left( 1 - 2 \sin (\theta/2) \right)$ and constrained Yao-graphs with an odd number of cones ($m \geq 5$) have spanning ratio at most $1 / \left( 1 - 2 \sin (3\theta/8) \right)$. As is the case with constrained $\theta$-graphs, these bounds match the current upper bounds in the unconstrained setting, which implies that like in the unconstrained setting using more cones can make the spanning ratio worse. 
\end{abstract}

\section{Introduction}
A geometric graph $G$ is a weighted graph whose vertices are points in the plane and whose edges are line segments between pairs of points. Every edge is weighted by the Euclidean distance between its endpoints. The distance between two vertices $u$ and $v$ in $G$, denoted by $d_G(u, v)$, is defined as the sum of the weights of the edges along the shortest path between $u$ and $v$ in $G$. A subgraph $H$ of $G$ is a $t$-spanner of $G$ (for $t\geq 1$) if for each pair of vertices $u$ and $v$, $d_H(u, v) \leq t \cdot d_G(u, v)$. The smallest value $t$ for which $H$ is a $t$-spanner is the {\em spanning ratio} or {\em stretch factor}. The graph $G$ is referred to as the {\em underlying graph} of $H$. The spanning properties of various geometric graphs have been studied extensively in the literature (see \cite{BS11,NS-GSN-06} for a comprehensive overview of the topic). We look at two specific types of geometric spanners: Yao-graphs and $\theta$-graphs. 

Introduced independently by Flinchbaugh and Jones~\cite{FJ81} and Yao~\cite{Y82}, Yao-graphs partition the plane around each vertex into $m$ disjoint cones, each having aperture $\theta = 2 \pi/m$. The Yao-graph with $m$ cones (also denoted as the $Y_m$-graph) is constructed in the following way: for each cone of each vertex $u$, connect $u$ to the vertex $v$ that is closest to $u$. However, neither Flinchbaugh and Jones nor Yao proved that these graphs are spanners. To the best of our knowledge, the first such proof was given by Alth{\"o}fer~\etal~\cite{A93}, who proved that for every spanning ratio $t > 1$, there exists an $m$ such that the $Y_m$-graph is a $t$-spanner. It appears that a similar result was already known by that time, since Clarkson~\cite{C87} remarked in 1987 that the $Y_{12}$-graph is a $(1 + \sqrt{3})$-spanner, though without providing a proof or reference. 

In 2004,  Bose~\etal~\cite{BMNSZ04} provided a more precise bound on the spanning ratio. They showed that Yao-graphs with at least 9 cones have spanning ratio at most $1 / (\cos \theta - \sin \theta)$. This was later strengthened to show that Yao-graphs with at least 7 cones are $1 / \left( 1 - 2 \sin (\theta/2) \right)$-spanners~\cite{BDDOSSW12Arxiv}. Recently, Damian and Raudonis~\cite{DR12} showed that the $Y_6$-graph is a $17.64$-spanner, which was later improved to $5.8$~\cite{BBDFKORTVX14}. Bose~\etal~\cite{BDDOSSW12} showed that the $Y_4$-graph has spanning ratio at most $8 \sqrt{2} \cdot (26 + 23 \sqrt{2}) \approx 663$ and Barba~\etal~\cite{BBDFKORTVX14} showed that the $Y_5$-graph is a $\left( 2 + \sqrt{3} \right)$-spanner. In the same paper, they also improved the upper bound on the spanning ratio of Yao-graphs with an odd number of cones to $1 / \left( 1 - 2 \sin (3\theta/8) \right)$. On the other hand, when a Yao-graph has fewer than 4 cones, El~Molla~\cite{E09} showed that there is no constant $t$ such that it is a $t$-spanner. 

Similar to Yao-graphs, $\theta$-graphs also partition the plane around each vertex into $m$ disjoint cones, each having aperture $\theta = 2 \pi/m$. However, unlike in the case of Yao-graphs, the $\theta_m$-graph is constructed by connecting each vertex $u$ to the vertex whose projection along the bisector of the cone is closest to $u$. This construction was introduced independently by Clarkson~\cite{C87} and Keil~\cite{K88}. Ruppert and Seidel~\cite{RS91} showed that the spanning ratio of these graphs is at most $1/(1 - 2 \sin (\theta/2))$, when $\theta < \pi/3$, i.e. there are at least 7 cones. Recent results include a tight spanning ratio of $1 + 2 \sin(\theta/2)$ for $\theta$-graphs with $4k + 2$ cones, where $k \geq 1$ and integer, and improved upper bounds for the other three families of $\theta$-graphs~\cite{BCMRV16}. It was also shown that the $\theta_5$-graph is a spanner with spanning ratio at most $\sqrt{50 + 22 \sqrt{5}} \approx 9.960$~\cite{BMRV2015} and the $\theta_4$-graph is a spanner with spanning ratio at most $(1 + \sqrt{2}) \cdot (\sqrt{2} + 36) \cdot \sqrt{4 + 2 \sqrt{2}} \approx 237$~\cite{BBCRV2013}. Constructions similar to those for Yao-graphs show that $\theta$-graphs with fewer than 4 cones are not spanners. In fact, until recently it was not known that the $\theta_3$-graph is connected~\cite{ABBBKRTV2013}. 

Most of the research for both Yao- and $\theta$-graphs, however, has focused on constructing spanners where the underlying graph is the complete Euclidean geometric graph. We study this problem in a more general setting with the introduction of line segment {\em constraints}. Specifically, let $P$ be a set of points in the plane and let $S$ be a set of line segments between two vertices in $P$, called \emph{constraints}. The set of constraints is planar, i.e. no two constraints intersect properly. Two vertices $u$ and $v$ can see each other if and only if either the line segment $uv$ does not properly intersect any constraint or $u v$ is itself a constraint. If two vertices $u$ and $v$ can see each other, the line segment $uv$ is a \emph{visibility edge}. The \emph{visibility graph} of $P$ with respect to a set of constraints $S$, denoted $\Vis(P,S)$, has $P$ as vertex set and all visibility edges as edge set. In other words, it is the complete graph on $P$ minus all edges that properly intersect one or more constraints in $S$.

This setting has been studied extensively within the context of motion planning amid obstacles. Clarkson~\cite{C87} was one of the first to study this problem and showed how to construct a linear-sized $(1+\epsilon)$-spanner of $\Vis(P,S)$. Subsequently, Das~\cite{D97} showed how to construct a spanner of $\Vis(P,S)$ with constant spanning ratio and constant degree. The Constrained Delaunay Triangulation was shown to be a 2.42-spanner of $\Vis(P,S)$~\cite{BK06}. Recently, it was also shown that the constrained $\theta_6$-graph is a 2-spanner of $\Vis(P,S)$~\cite{BFRV12}. 

In this paper, we generalize the recent results on unconstrained $\theta$-graphs by Bose~\etal~\cite{BCMRV16} to the constrained setting. There are two main obstacles that differentiate this work from previous results. First, the main difficulty with the constrained setting is that induction cannot be applied directly, as the destination need not be visible from the vertex closest to the source (see Figure~\ref{fig:ConvexChain}, where $w$ is not visible from $v_0$, the vertex closest to $u$). Second, when the graph does not have $4k + 2$ cones, the cones do not line up as nicely as in~\cite{BFRV12}, making it more difficult to apply induction. 

We overcome these two difficulties and show that constrained $\theta$-graphs with $4k + 2$ cones have a spanning ratio of at most $1 + 2 \sin(\theta/2)$, where $\theta$ is $2 \pi / (4k + 2)$. Since the lower bounds of the unconstrained $\theta$-graphs carry over to the constrained setting, this shows that this spanning ratio is tight. We also show that constrained $\theta$-graphs with $4k + 4$ cones have a spanning ratio of at most $1 + 2 \sin(\theta/2) / (\cos(\theta/2) - \sin(\theta/2))$, where $\theta$ is $2 \pi / (4k + 4)$. Finally, we show that constrained $\theta$-graphs with $4k + 3$ or $4k + 5$ cones have a spanning ratio of at most $\cos (\theta/4) / (\cos (\theta/2) - \sin (3\theta/4))$, where $\theta$ is $2 \pi / (4k + 3)$ or $2 \pi / (4k + 5)$. 

\begin{table}[ht]
  \begin{center}
    \begin{tabular}{| >{\centering\arraybackslash}m{\dimexpr.14\linewidth-2\tabcolsep} || >{\centering\arraybackslash}m{\dimexpr.3\linewidth-2\tabcolsep} | >{\centering\arraybackslash}m{\dimexpr.3\linewidth-2\tabcolsep} |}
    \hline
    $m$ & $\theta_m$-Graph & $Y_m$-Graph \\ 
    \hline \hline
     4 & ? & \vspace{2ex} ? \\ [2ex]
    \hline
    5 & ? & \vspace{2ex} $\frac{1}{1 - 2 \sin \left( \frac{3\theta}{8} \right)} \approx 10.87$ \\ [2ex]
    \hline
    6 & $1 + 2 \sin \left( \frac{\theta}{2} \right) = 2$~\cite{BFRV12} & \vspace{2ex} ? \\ [2ex]
    \hline
    $4k + 2$ ($k \geq 2$) & $1 + 2 \sin \left( \frac{\theta}{2} \right)$ & \vspace{2ex} $\frac{1}{1 - 2 \sin \left( \frac{\theta}{2} \right)}$ \\ [2ex]
    \hline
    $4k + 3$ ($k \geq 1$) & $\frac{\cos \left( \frac{\theta}{4} \right)}{\cos \left( \frac{\theta}{2} \right) - \sin \left( \frac{3\theta}{4} \right)}$ & \vspace{2ex} $\frac{1}{1 - 2 \sin \left( \frac{3\theta}{8} \right)}$ \\ [2ex] 
    \hline
    $4k + 4$ ($k \geq 1$) & $1 + \frac{2 \sin \left( \frac{\theta}{2} \right)}{\cos \left( \frac{\theta}{2} \right) - \sin \left( \frac{\theta}{2} \right)}$ & \vspace{2ex} $\frac{1}{1 - 2 \sin \left( \frac{\theta}{2} \right)}$ \\ [2ex]
    \hline
    $4k + 5$ ($k \geq 1$) & $\frac{\cos \left( \frac{\theta}{4} \right)}{\cos \left( \frac{\theta}{2} \right) - \sin \left( \frac{3\theta}{4} \right)}$ & \vspace{2ex} $\frac{1}{1 - 2 \sin \left( \frac{3\theta}{8} \right)}$ \\ [2ex] 
    \hline
    \end{tabular}
  \end{center} 
  \caption{An overview of the upper bounds on the spanning ratios of constrained $\theta$-graphs and Yao-graphs}
  \label{tab:UpperBounds}
\end{table}

\begin{table}[h!]
  \begin{center}
    \begin{tabular}{| >{\centering\arraybackslash}m{\dimexpr.14\linewidth-2\tabcolsep} || >{\centering\arraybackslash}m{\dimexpr.48\linewidth-2\tabcolsep} | >{\centering\arraybackslash}m{\dimexpr.36\linewidth-2\tabcolsep} |}
    \hline
    $m$ & $\theta_m$-Graph & $Y_m$-Graph \\ 
    \hline \hline
     4 &  7~\cite{BBCRV2013} & \vspace{2ex} 3.89~\cite{darryl} \\ [2ex]
    \hline
    5 & $\frac{1}{2}(11\sqrt{5} - 17) \approx 3.79$~\cite{BMRV2015} & \vspace{2ex} 2.87~\cite{BBDFKORTVX14} \\ [2ex]
    \hline
    6 &  $1 + 2 \sin \left( \frac{\theta}{2} \right) = 2$~\cite{BCMRV16} & \vspace{2ex} 2~\cite{BBDFKORTVX14} \\ [2ex]
    \hline
    $4k + 2$ ($k \geq 2$) & $1 + 2 \sin \left( \frac{\theta}{2} \right)$~\cite{BCMRV16} & \vspace{2ex} $1 + 2 \sin\left(\frac{\theta}{2}\right)$~\cite{BBDFKORTVX14} \\ [2ex]
    \hline
    $4k + 3$ ($k \geq 1$) & $\frac{3\cos\left(\frac{\theta}{4}\right)+\cos\left(\frac{3\theta}{4}\right)+\sin\left(\frac{\theta}{2}\right)+\sin \theta +\sin\left(\frac{3\theta}{2}\right)}{3\cos\left(\frac{\theta}{2}\right)+\cos\left(\frac{3\theta}{2}\right)}$~\cite{BCMRV16} & \vspace{2ex} $1 + 2 \sin \left( \frac{3\theta}{8} \right) + g(\theta)$~\cite{BBDFKORTVX14} \\ [2ex] 
    \hline
    $4k + 4$ ($k \geq 1$) & $1 + 2 \tan \left( \frac{\theta}{2} \right) + 2 \tan^2 \left( \frac{\theta}{2} \right)$~\cite{BCMRV16} & \vspace{2ex} $1 + 2 \sin \left( \frac{\theta}{2} \right) \left( 1 + \tan \left( \frac{\theta}{2} \right) \right)$~\cite{BBDFKORTVX14} \\ [2ex]
    \hline
    $4k + 5$ ($k \geq 1$) & $f(\theta) + \tan\left(\frac{\theta}{2}\right) + \frac{1}{2}\sec\left(\frac{\theta}{2}\right)\tan\left(\frac{\theta}{2}\right)$~\cite{BCMRV16} & \vspace{2ex} $1 + 2 \sin \left( \frac{3\theta}{8} \right) + 4 \sin \left( \frac{5\theta}{16} \right) \sin \left( \frac{3\theta}{8} \right)$~\cite{BBDFKORTVX14} \\ [2ex] 
    \hline
    \end{tabular}
  \end{center} 
  \caption{An overview of the lower bounds on the spanning ratios of constrained $\theta$-graphs and Yao-graphs, where $f(\theta) = \frac{1}{2} \sqrt{4\sec\left(\frac{\theta}{2}\right) + 7\sec^2\left(\frac{\theta}{2}\right) + 4\sec^3\left(\frac{\theta}{2}\right) + \sec^4\left(\frac{\theta}{2}\right) - 8\cos\left(\frac{\theta}{2}\right) - 4}$ and $g(\theta) =  4 \frac{\left( \sin \left( \frac{13\theta}{16} \right) + \sin \left( \frac{19\theta}{16} \right)\right) \sin \left( \frac{\theta}{16} \right) \sin \left( \frac{3\theta}{8} \right)}{\sin(2\theta)}$}
  \label{tab:LowerBounds}
\end{table}

Furthermore, to the best of our knowledge, Yao-graphs have not been considered in the constrained setting. As such, it is unknown whether they are spanners of $\Vis(P,S)$. In this paper, we set an important first step towards answering this question by showing that constrained Yao-graphs with 5 or at least 7 cones are spanners. In particular, we prove that constrained Yao-graphs with at least 7 cones have spanning ratio at most $1 / \left( 1 - 2 \sin (\theta/2) \right)$. When the constrained Yao-graph has an odd number of cones, we can improve on this result, and extend it to the Yao-graph with 5 cones, and show an upper bound of $1 / \left( 1 - 2 \sin (3\theta/8) \right)$. Surprisingly, these bounds match the current upper bounds in the unconstrained setting. An overview of the upper bounds for constrained $\theta$-graphs and Yao-graphs can be found in Table~\ref{tab:UpperBounds}. 

Finally, since the lower bounds for the unconstrained setting also hold in the constrained setting, this also implies that even in the presence of constraints, using more cones can make the spanning ratio worse. An overview of the current lower bounds for both $\theta$-graphs and Yao-graphs can be found in Table~ \ref{tab:LowerBounds}.

\section{Preliminaries}
We define a \emph{cone} $C$ to be the region in the plane between two rays originating from a vertex referred to as the apex of the cone. When constructing a (constrained) $\theta_m$- or $Y_m$-graph, for each vertex $u$ consider the rays originating from $u$ with the angle between consecutive rays being $\theta = 2 \pi / m$. Each pair of consecutive rays defines a cone. The cones are oriented such that the bisector of one cone coincides with the vertical halfline through $u$ that lies above $u$. Let this cone of $u$ be $C_0$ and number the cones in clockwise order around $u$ (see Figure~\ref{fig:Cones}). The cones around the other vertices have the same orientation as the ones around $u$. We write $C_i^u$ to indicate the $i$-th cone of a vertex $u$. 

\begin{figure}[ht]
  \center
  \begin{minipage}[t]{0.45\textwidth}
    \begin{center}
      \includegraphics{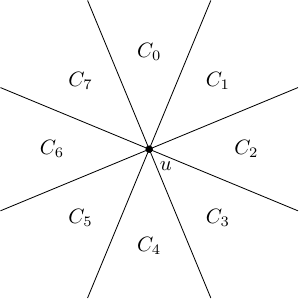}
    \end{center}
    \caption{The cones having apex $u$ in the $\theta_8$- and $Y_8$-graph}
    \label{fig:Cones}
  \end{minipage}
  \hspace{0.05\linewidth}
  \begin{minipage}[t]{0.45\textwidth}
    \begin{center}
      \includegraphics{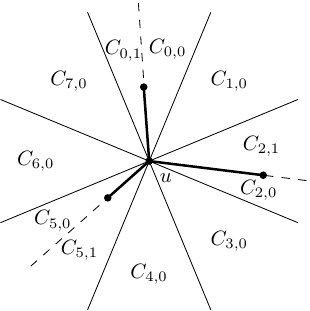}
    \end{center}
    \caption{The subcones having apex $u$ in the constrained $\theta_8$- and $Y_8$-graph. Constraints are shown as thick segments}
    \label{fig:ConstrainedCones}
  \end{minipage}
\end{figure}

Let vertex $u$ be an endpoint of a constraint $c$ and let the other endpoint $v$ lie in cone $C_i^u$. The lines through all such constraints $c$ split $C_i^u$ into several \emph{subcones}. We use $C_{i, j}^u$ to denote the $j$-th subcone of $C_i^u$ (see Figure~\ref{fig:ConstrainedCones}). When a constraint $c = (u, v)$ splits a cone of $u$ into two subcones, we assume that $v$ lies in both of these subcones. We consider a cone that is not split to be a single subcone. For ease of exposition, we only consider point sets in general position: no two points lie on a line parallel to one of the rays that define the cones, no two points lie on a line perpendicular to the bisector of a cone, and no three points are collinear. 

We now introduce the constrained $Y_m$-graph: for each subcone $C_{i, j}$ of each vertex $u$, add an edge from $u$ to the closest vertex in that subcone that can see $u$ (see Figure~\ref{fig:ClosestYao}). When there exist multiple closest vertices in a subcone, we add an edge to only one of them. More formally, we add an edge between two vertices $u$ and $v$ if $v$ can see $u$, $v \in C_{i, j}^u$, and for all points $w \in C_{i, j}^u$ that can see $u$, $|u v| \leq |u w|$, where $|x y|$ denotes the length of the line segment between two points $x$ and $y$ and ties are broken arbitrarily. 

\begin{figure}[ht]
  \center
  \begin{minipage}[t]{0.45\textwidth}
    \begin{center}
      \includegraphics{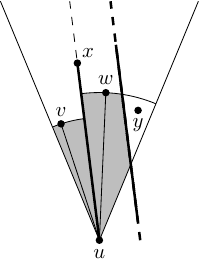}
    \end{center}
    \caption{Vertices $v$ and $w$ are the closest visible vertices to $u$ in the left and right subcone of the constrained Yao-graph, where $u x$ is a constraint}
    \label{fig:ClosestYao}
  \end{minipage}
  \hspace{0.05\linewidth}
  \begin{minipage}[t]{0.45\textwidth}
    \begin{center}
      \includegraphics{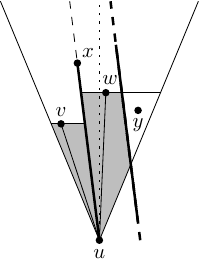}
    \end{center}
    \caption{Vertices $v$ and $w$ are the closest visible vertices to $u$ in the left and right subcone of the constrained $\theta$-graph, where $u x$ is a constraint}
    \label{fig:ClosestTheta}
  \end{minipage}
\end{figure}

The constrained $\theta_m$-graph is similar to the constrained $Y_m$-graph, but uses a different method to determine which vertex is closest to a vertex $u$: for each subcone $C_{i, j}$ of each vertex $u$, add an edge from $u$ to the closest vertex in that subcone that can see $u$, where distance is measured along the bisector of the original cone (\emph{not the subcone}, see Figure~\ref{fig:ClosestTheta}). More formally, we add an edge between two vertices $u$ and $v$ if $v$ can see $u$, $v \in C_{i, j}^u$, and for all points $w \in C_{i, j}^u$ that can see $u$, $|u v'| \leq |u w'|$, where $v'$ and $w'$ denote the projection of $v$ and $w$ on the bisector of $C_i^u$ and $|x y|$ denotes the length of the line segment between two points $x$ and $y$. Note that our assumption of general position implies that each vertex adds at most one edge for each of its subcones. 

Finally, we define the notion of a \emph{canonical triangle} for constrained $\theta$-graphs. Given a vertex $w$ in the cone $C_i$ of vertex $u$, we define the \emph{canonical triangle} \canon{u}{w} to be the triangle defined by the borders of $C_i^u$ and the line through $w$ perpendicular to the bisector of $C_i^u$. Note that subcones do not define canonical triangles. We use $\alpha$ to denote the unsigned angle between $u w$ and the bisector of $C_i^u$ (see Figure~\ref{fig:CanonicalTriangle}). Note that for any pair of vertices $u$ and $w$, there exist two canonical triangles: \canon{u}{w} and \canon{w}{u}. We say that a region is \emph{empty} if it does not contain any vertex of $P$.

 \begin{figure}[ht]
   \begin{center}
     \includegraphics{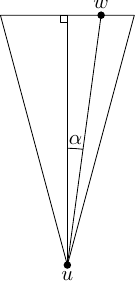}
   \end{center}
   \caption{The canonical triangle \canon{u}{w}}
   \label{fig:CanonicalTriangle}
 \end{figure}

\subsection{Some Useful Lemmas}
In this section, we list a number of lemmas that are used when bounding the spanning ratio of the various graphs. Note that these lemmas are not new, as they are already used in~\cite{BFRV12,BCMRV16}, though some are expanded to work for all four families of constrained $\theta$-graphs. Though the following lemma was applied to constrained $\theta$-graphs in~\cite{BFRV12}, the property holds for any visibility graph. To avoid confusion, we explicitly define a region to be \emph{empty} if it does not contain any vertex of $P$. 

\begin{lemma}
  \label{lem:ConvexChain}
  Let $u$, $v$, and $w$ be three arbitrary points in the plane such that $u w$ and $v w$ are visibility edges and $w$ is not the endpoint of a constraint intersecting the interior of triangle $u v w$. Then there exists a convex chain of visibility edges (different from the chain consisting of $uw$ and $wv$) from $u$ to $v$ in triangle $u v w$, such that the polygon defined by $u w$, $w v$ and the convex chain is empty and does not contain any constraints.
\end{lemma}

 \begin{figure}[ht]
   \begin{center}
     \includegraphics{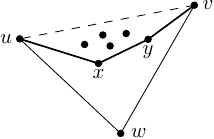}
   \end{center}
   \caption{The convex chain between vertices $u$ and $v$, where thick lines are visibility edges}
   \label{fig:VisiblePointInsideTriangle}
 \end{figure}

Next, we use two lemmas from \cite{BCMRV16} to bound the length of certain line segments. We use $\angle xyz$ to denote the smaller angle between line segments $xy$ and $yz$. 

\begin{lemma}
  \label{lem:ApplyFourPoints} 
  Let $u$, $v$ and $w$ be three vertices in the \graph{x}, $k \geq 1$ and $x \in \{2, 3, 4, 5\}$, such that $w \in C_0^u$ and $v \in \canon{u}{w}$, is to the left of $uw$. Let $a$ be the intersection of the side of $\canon{u}{w}$ opposite $u$ and the left boundary of $C_0^v$. Let $C_i^v$ denote the cone of $v$ that contains $w$ and let $c$ and $d$ be the upper and lower corner of $\canon{v}{w}$. If $1 \leq i \leq k-1$, or $i = k$ and $|c w| \leq |d w|$, then $\max \left\{|v c| + |c w|, |v d| + |d w|\right\} \leq |v a| + |a w|$ and $\max \left\{|c w|, |d w|\right\} \leq |a w|$.
\end{lemma}

\begin{figure}[ht]
  \center
  \begin{minipage}[b]{0.4\textwidth}
    \begin{center}
      \includegraphics{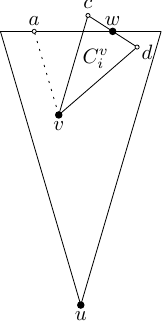}
    \end{center}
    \caption{The situation where we \newline apply Lemma~\ref{lem:ApplyFourPoints}}
    \label{fig:ApplyQuadrilateral}
  \end{minipage}
  \hspace{0.05\linewidth}
  \begin{minipage}[b]{0.5\textwidth}
    \begin{center}
      \includegraphics{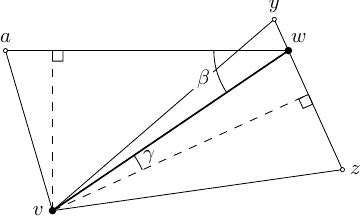}
    \end{center}
    \caption{The situation where we \newline apply Lemma~\ref{lem:CalculationCase}}
    \label{fig:CalculationLemma}
  \end{minipage}
\end{figure}

\begin{lemma}
  \label{lem:CalculationCase}
  Let $u$, $v$ and $w$ be three vertices in the \graph{x}, $x \in \{2, 3, 4, 5\}$, such that $w \in C_0^u$, $v \in \canon{u}{w}$ to the left of $uw$, and $w \not \in C_0^v$. Let $a$ be the intersection of the side of $\canon{u}{w}$ opposite $u$ and the line through $v$ parallel to the left boundary of \canon{u}{w}. Let $y$ and $z$ be the corners of $\canon{v}{w}$ opposite to $v$. Let $\beta = \angle a w v$ and let $\gamma$ be the unsigned angle between $v w$ and the bisector of \canon{v}{w}. Let \const be a positive constant. If $\const \geq \frac{\cos \gamma - \sin \beta}{\cos \left( \frac{\theta}{2} - \beta \right) - \sin \left( \frac{\theta}{2} + \gamma \right)}$, then $|v p| + \const \cdot |p w| \leq |v a| + \const \cdot |a w|$, where $p$ is $y$ if $|y w| \geq |z w|$ and $p$ is $z$ if $|y w| < |z w|$.
\end{lemma}

\section[Constrained $\theta$-Graphs]{Constrained $\boldsymbol{\theta}$-Graphs}
In this section, we provide tight bounds on the spanning ratio for the constrained \graph{2} and upper bounds on those for the constrained \graph{3}, the constrained \graph{4}, the constrained \graph{5}. For the latter three families, we provide a generic framework for the upper bound on the spanning ratio, to avoid having to prove the same statements for each of the families individually.

\subsection[The Constrained \Graph{2}]{The Constrained $\boldsymbol{\theta_{(4 k + 2)}}$-Graph}
In this section we prove that the constrained \graph{2} has spanning ratio at most $1 + 2 \cdot \sin (\theta/2)$. Since this is also a lower bound~\cite{BCMRV16}, this proves that this spanning ratio is tight. 

\begin{theorem}
\label{theo:PathLength4k+2Constrained}
  Let $u$ and $w$ be two vertices in the plane such that $u$ can see $w$. Let $m$ be the midpoint of the side of \canon{u}{w} opposing $u$ and let $\alpha$ be the unsigned angle between $u w$ and $u m$. There exists a path connecting $u$ and $w$ in the constrained \graph{2} of length at most \[\left( \left(\frac{1 + \sin \left(\frac{\theta}{2}\right)}{\cos \left(\frac{\theta}{2}\right)} \right) \cdot \cos \alpha + \sin \alpha \right) \cdot |u w|.\] 
\end{theorem}
\begin{proof}
  We assume without loss of generality that $w \in C_0^u$. We prove the theorem by induction on the area of $\canon{u}{w}$. Formally, we perform induction on the rank, when ordered by area, of the triangles \canon{x}{y} for all pairs of vertices $x$ and $y$ that can see each other. Let $a$ and $b$ be the upper left and right corner of $\canon{u}{w}$, and let $A$ and $B$ be the triangles $u a w$ and $u b w$ (see Figure~\ref{fig:ConvexChain}).

  \begin{figure}[ht]
    \begin{center}
      \includegraphics{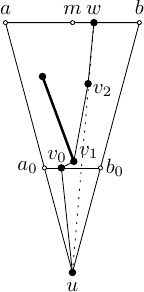}
    \end{center}
    \caption{A convex chain from $v_0$ to $w$}
    \label{fig:ConvexChain}
  \end{figure}

  Our inductive hypothesis is the following, where $\delta(u,w)$ denotes the length of the shortest path from $u$ to $w$ in the constrained \graph{2}:
  \begin{itemize}
    \item If $A$ is empty, then $\delta(u, w) \leq |u b| + |b w|$.
    \item If $B$ is empty, then $\delta(u, w) \leq |u a| + |a w|$.
    \item If neither $A$ nor $B$ is empty, then $\delta(u, w) \leq \max\{|u a| + |a w|, |u b| + |b w|\}$. 
  \end{itemize}

  We first show that this induction hypothesis implies the theorem: $|u m| = |u w| \cdot \cos \alpha$, $|m w| = |u w| \cdot \sin \alpha$, $|a m| = |b m| = |u w| \cdot \cos \alpha \cdot \tan (\theta/2)$, and $|u a| = |u b| = |u w| \cdot \cos \alpha / \cos (\theta/2)$. Thus the induction hypothesis gives that \[\delta(u, w) \leq|u a| + |a m| + |m w| = \left( \left(\frac{1 + \sin \left(\frac{\theta}{2}\right)}{\cos \left(\frac{\theta}{2}\right)} \right) \cdot \cos \alpha + \sin \alpha \right) \cdot |u w|.\] 

We now return our attention to proving that the induction hypothesis holds. 

  \textbf{Base case:} $\canon{u}{w}$ has rank 1. Since the triangle is a smallest triangle such that $u$ and $w$ can see each other, $w$ is the closest visible vertex to $u$ in that cone. Hence the edge $u w$ is part of the constrained \graph{2}, and $\delta(u, w) = |u w|$. From the triangle inequality, we have $|u w| \leq \min\{|u a| + |a w|, |u b| + |b w|\}$, so the induction hypothesis holds.

  \textbf{Induction step:} We assume that the induction hypothesis holds for all pairs of vertices that can see each other and have a canonical triangle whose area is smaller than the area of $\canon{u}{w}$. 

  If $u w$ is an edge in the constrained \graph{2}, the induction hypothesis follows by the same argument as in the base case. If there is no edge between $u$ and $w$, let $v_0$ be the closest visible vertex to $u$ in the subcone of $u$ that contains $w$, and let $a_0$ and $b_0$ be the upper left and right corner of $\canon{u}{v_0}$ (see Figure~\ref{fig:ConvexChain}). By definition, $\delta(u, w) \leq |u v_0| + \delta(v_0, w)$, and by the triangle inequality, $|u v_0| \leq \min\{|u a_0| + |a_0 v_0|, |u b_0| + |b_0 v_0|\}$. We assume without loss of generality that $v_0$ lies to the left of $u w$, which means that $A$ is not empty.

  Since $u w$ and $u v_0$ are visibility edges, by applying Lemma~\ref{lem:ConvexChain} to triangle $v_0 u w$,  a convex chain $v_0, ..., v_l = w$ of visibility edges connecting $v_0$ and $w$ exists (see Figure~\ref{fig:ConvexChain}). Note that, since $v_0$ is the closest visible vertex to $u$, every vertex along the convex chain lies above the horizontal line through $v_0$. 
      
  We now look at two consecutive vertices $v_{j-1}$ and $v_j$ along the convex chain. There are four types of configurations (see Figure~\ref{fig:Configurations}): \mbox{(i) $v_j \in C_k^{v_{j-1}}$,} \mbox{(ii) $v_j \in C_i^{v_{j-1}}$} where \mbox{$1 \leq i < k$,} (iii) $v_j \in C_0^{v_{j-1}}$ and $v_j$ lies to the right of or has the same $x$-coordinate as $v_{j-1}$, and (iv) $v_j \in C_0^{v_{j-1}}$ and $v_j$ lies to the left of $v_{j-1}$. By convexity, the direction of $\overrightarrow{v_j v_{j+1}}$ is rotating counterclockwise for increasing $j$. Thus, these configurations occur in the order Type (i), Type (ii), Type (iii), and Type (iv) along the convex chain from $v_0$ to $w$. We bound $\delta(v_{j-1}, v_j)$ as follows:

  \textbf{Type (i):} If $v_j \in C_k^{v_{j-1}}$, let $a_j$ and $b_j$ be the upper and lower left corners of \canon{v_j}{v_{j-1}} and let $B_j = v_{j-1} b_j v_j$. Note that since $v_j \in C_k^{v_{j-1}}$, $a_j$ is also the intersection of the left boundary of $C_0^{v_{j-1}}$ and the horizontal line through $v_j$. We note that triangle $B_j$ is contained in the area defined by the convex chain, $u v_0$, and $u w$, since all three vertices of $B_j$ lie to the left of $u w$, below the line through $v_{j-1} v_j$, and to the right of the line through $u v_{j-1}$. Hence, triangle $B_j$ must be empty. Since $v_j$ can see $v_{j-1}$ and \canon{v_j}{v_{j-1}} has smaller area than \canon{u}{w}, the induction hypothesis gives that $\delta(v_{j-1}, v_j)$ is at most $|v_{j-1} a_j| + |a_j v_j|$.

  \textbf{Type (ii):} If $v_j \in C_i^{v_{j-1}}$ where $1 \leq i < k$, let $c$ and $d$ be the upper and lower right corner of \canon{v_{j-1}}{v_j}. Let $a_j$ be the intersection of the left boundary of $C_0^{v_{j-1}}$ and the horizontal line through $v_j$. Since $v_j$ can see $v_{j-1}$ and \canon{v_{j-1}}{v_j} has smaller area than \canon{u}{w}, the induction hypothesis gives that $\delta(v_{j-1}, v_j)$ is at most $\max\{|v_{j-1} c| + |c v_j|, |v_{j-1} d| + |d v_j|\}$. Since $v_j \in C_i^{v_{j-1}}$ where $1 \leq i < k$, we can apply Lemma~\ref{lem:ApplyFourPoints} (where $v$, $w$, and $a$ from Lemma~\ref{lem:ApplyFourPoints} are $v_{j-1}$, $v_j$, and $a_j$), which gives us that $\max\{|v_{j-1} c| + |c v_j|, |v_{j-1} d| + |d v_j|\} \leq |v_{j-1} a_j| + |a_j v_j|$. 
 
  \vspace{-1em}
  \begin{figure}[ht]
    \begin{center}
      \includegraphics{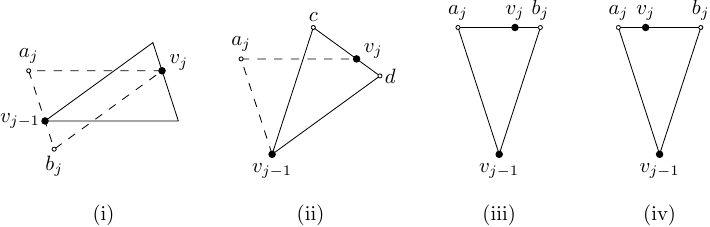}
    \end{center}
    \vspace{-1em}
    \caption{The four types of configurations}
    \label{fig:Configurations}
  \end{figure}
  \vspace{-1em}
 
  \textbf{Type (iii):} If $v_j \in C_0^{v_{j-1}}$ and $v_j$ lies to the right of or has the same $x$-coordinate as $v_{j-1}$, let $a_j$ and $b_j$ be the left and right corners of \canon{v_{j-1}}{v_j} and let $A_j = v_{j-1} a_j v_j$ and $B_j = v_{j-1} b_j v_j$. Since $v_j$ can see $v_{j-1}$ and \canon{v_{j-1}}{v_j} has smaller area than \canon{u}{w}, we can apply the induction hypothesis. Regardless of whether $A_j$ and $B_j$ are empty or not, $\delta(v_{j-1}, v_j)$ is at most $\max\{|v_{j-1} a_j| + |a_j v_j|, |v_{j-1} b_j| + |b_j v_j|\}$. Since $v_j$ lies to the right of or has the same $x$-coordinate as $v_{j-1}$, we know that $|v_{j-1} a_j| + |a_j v_j| \geq |v_{j-1} b_j| + |b_j v_j|$, so $\delta(v_{j-1}, v_j)$ is at most $|v_{j-1} a_j| + |a_j v_j|$.

  \textbf{Type (iv):} If $v_j \in C_0^{v_{j-1}}$ and $v_j$ lies to the left of $v_{j-1}$, let $a_j$ and $b_j$ be the left and right corners of \canon{v_{j-1}}{v_j} and let $A_j = v_{j-1} a_j v_j$ and $B_j = v_{j-1} b_j v_j$. Since $v_j$ can see $v_{j-1}$ and \canon{v_{j-1}}{v_j} has smaller area than \canon{u}{w}, we can apply the induction hypothesis. Thus, if $B_j$ is empty, $\delta(v_{j-1}, v_j)$ is at most $|v_{j-1} a_j| + |a_j v_j|$ and if $B_j$ is not empty, $\delta(v_{j-1}, v_j)$ is at most $|v_{j-1} b_j| + |b_j v_j|$.

  Now that we have bounded the length of the inductive path for each type of configuration, we use these configurations to bound the total length of the path. We consider three cases: (a) $\angle a w u \leq \pi/2$, (b) $\angle a w u > \pi/2$ and $B$ is empty, and (c) $\angle a w u > \pi/2$ and $B$ is not empty. 

  \textbf{Case (a):} If $\angle a w u \leq \pi/2$, the convex chain cannot contain any Type (iv) configurations: for Type~(iv) configurations to occur, $v_j$ needs to lie to the left of $v_{j-1}$. However, by construction, $v_j$ lies on or to the right of the line through $v_{j-1}$ and $w$. Hence, since $\angle a w v_{j-1} < \angle a w u \leq \pi/2$, $v_j$ lies to the right of or has the same $x$-coordinate as $v_{j-1}$. We can now bound $\delta(u, w)$ by using these bounds: 
  \begin{eqnarray*}
    \delta(u, w) &\leq& |u v_0| + \sum_{j=1}^l \delta(v_{j-1}, v_j) \\
		 &\leq& |u a_0| + |a_0 v_0| + \sum_{j=1}^l (|v_{j-1} a_j| + |a_j v_j|) \\
		 &=& |u a| + |a w|.
  \end{eqnarray*}

  \textbf{Case (b):} If $\angle a w u > \pi/2$ and $B$ is empty, the convex chain can contain Type~(iv) configurations. However, since $B$ is empty and the area between the convex chain and $u w$ is empty (by Lemma~\ref{lem:ConvexChain}), all $B_j$ are also empty. Using the computed bounds on the lengths of the paths between the points along the convex chain, we can bound $\delta(u, w)$ as in the previous case.

  \textbf{Case (c):} If $\angle a w u > \pi/2$ and $B$ is not empty, the convex chain can contain Type~(iv) configurations and since $B$ is not empty, the triangles $B_j$ need not be empty. Recall that $v_0$ lies in $A$, hence neither $A$ nor $B$ is empty. Therefore, it suffices to prove that $\delta(u, w) \leq \max\{|ua| + |aw|, |ub| + |bw|\} = |ub| + |bw|$. Let \canon{v_{j'}}{v_{j'+1}} be the first Type (iv) configuration along the convex chain (if it has any), let $a'$ and $b'$ be the upper left and right corner of \canon{u}{v_{j'}}, and let $b''$ be the upper right corner of \canon{v_{j'}}{w}. We can bound $\delta(u, w)$ as follows (see Figure~\ref{fig:SpanningProofCase3-4k+2}):

  \begin{figure}[ht]
    \begin{center}
      \includegraphics{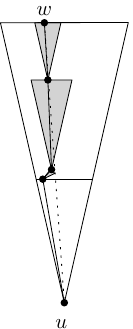}
      \hspace{0.1em}
      \includegraphics{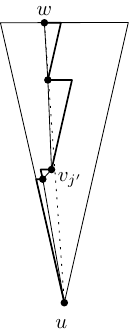}
      \hspace{0.1em}
      \includegraphics{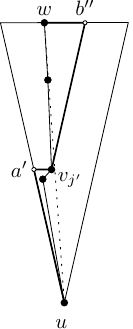}
      \hspace{0.1em}
      \includegraphics{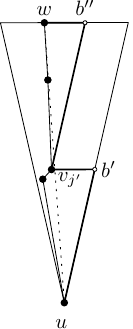}
      \hspace{0.1em}
      \includegraphics{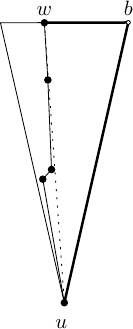}
    \end{center}
    \caption{Visualization of the paths (thick lines) in the inequalities of case (c)}
    \label{fig:SpanningProofCase3-4k+2}
  \end{figure}  
  
  \begin{eqnarray*}
    \delta(u, w) &\leq& |u v_0| + \sum_{j=1}^l \delta(v_{j-1}, v_j) \\
		 &\leq& |u a_0| + |a_0 v_0| + \sum_{j=1}^{j'} (|v_{j-1} a_j| + |a_j v_j|) + \sum_{j=j'+1}^l (|v_{j-1} b_j| + |b_j v_j|) \\
		 &=& |u a'| + |a' v_{j'}| + |v_{j'} b''| + |b'' w| \\
		 &\leq& |u b'| + |b' v_{j'}| + |v_{j'} b''| + |b'' w| \\
		 &=& |u b| + |b w|. 
  \end{eqnarray*} 
\end{proof}

Since $\left( \left(1 + \sin (\theta/2) \right)/\cos (\theta/2) \right) \cdot \cos \alpha + \sin \alpha$ is increasing in $\alpha$, for $\alpha \in [0, \theta/2]$ and fixed $\theta \in [0, \pi/3]$, it is maximized when $\alpha = \theta/2$, and we obtain the following corollary: 

\begin{corollary}
  \label{cor:SpanningRatioConstrained}
  The constrained \graph{2} is a $\left( 1 + 2 \cdot \sin \left( \frac{\theta}{2} \right) \right)$-spanner of $\Vis(P,S)$. 
\end{corollary}

\subsection{Generic Framework for the Spanning Proof}
Next, we modify the spanning proof from the previous section and provide a generic framework for the spanning proof for the other three families of $\theta$-graphs. After providing this framework, we complete the proofs for the individual families. 

The general inductive approach used in this framework is similar to that used in the proof of Theorem~\ref{theo:PathLength4k+2Constrained}. However, since for these three remaining families the line perpendicular to the bisector of the cone is not parallel to a cone boundary, the induction hypothesis  needs to be modified. While this modification does not preserve the tightness of the bound on the spanning ratio, it does allow us to make the proof more generic, hence leading to the framework that works for all three families. 

\begin{theorem}
\label{theo:PathLengthGenericConstrained}
  Let $u$ and $w$ be two vertices in the plane such that $u$ can see $w$. Let $m$ be the midpoint of the side of \canon{u}{w} opposing $u$ and let $\alpha$ be the unsigned angle between $u w$ and $u m$. There exists a path connecting $u$ and $w$ in the constrained \graph{x} of length at most \[\left( \frac{\cos \alpha}{\cos \left(\frac{\theta}{2}\right)} + \const \cdot \left(\cos \alpha \cdot \tan \left(\frac{\theta}{2}\right) + \sin \alpha\right) \right) \cdot |u w|,\] where $\const \geq 1$ is a function that depends on $x \in \{3, 4, 5\}$ and $\theta$. For the \graph{4}, \const is at most $1 / (\cos (\theta/2) - \sin (\theta/2))$ and for the \graph{3} and \graph{5}, \const is at most $\cos (\theta/4) /$ $(\cos (\theta/2) - \sin (3\theta/4))$.
\end{theorem}
\begin{proof} 
  We prove the theorem by induction on the area of $\canon{u}{w}$. Formally, we perform induction on the rank, when ordered by area, of the triangles \canon{x}{y} for all pairs of vertices $x$ and $y$ that can see each other. We assume without loss of generality that $w \in C_0^u$. Let $a$ and $b$ be the upper left and right corner of $\canon{u}{w}$ (see Figure~\ref{fig:ConvexChain}).

  Our inductive hypothesis is the following, where $\delta(u,w)$ denotes the length of the shortest path from $u$ to $w$ in the constrained \graph{x}: $\delta(u, w) \leq \max\{|u a| + \const \cdot |a w|, |u b| + \const \cdot |b w|\}$. 

  We first show that this induction hypothesis implies the theorem: $|u m| = |u w| \cdot \cos \alpha$, $|m w| = |u w| \cdot \sin \alpha$, $|a m| = |b m| = |u w| \cdot \cos \alpha \cdot \tan (\theta/2)$, and $|u a| = |u b| = |u w| \cdot \cos \alpha / \cos (\theta/2)$. Thus the induction hypothesis gives that \[\delta(u, w) \leq |u a| + \const \cdot (|a m| + |m w|) = \left( \frac{\cos \alpha}{\cos \left(\frac{\theta}{2}\right)} + \const \cdot \left(\cos \alpha \tan \left(\frac{\theta}{2}\right) + \sin \alpha\right) \right) \cdot |u w|.\]

We now return our attention to proving that the induction hypothesis holds. 

  \textbf{Base case:} $\canon{u}{w}$ has rank 1. Since the triangle is a smallest triangle such that $u$ and $w$ can see each other, $w$ is the closest visible vertex to $u$ in that cone. Hence the edge $u w$ is part of the constrained \graph{x}, and $\delta(u, w) = |u w|$. From the triangle inequality and the fact that $\const \geq 1$, we have $|u w| \leq \min\{|u a| + \const \cdot |a w|, |u b| + \const \cdot |b w|\}$, so the induction hypothesis holds.

  \textbf{Induction step:} We assume that the induction hypothesis holds for all pairs of vertices that can see each other and have a canonical triangle whose area is smaller than the area of $\canon{u}{w}$. 

  If $u w$ is an edge in the constrained \graph{x}, the induction hypothesis follows by the same argument as in the base case. If there is no edge between $u$ and $w$, let $v_0$ be the closest visible vertex to $u$ in the subcone of $u$ that contains $w$, and let $a_0$ and $b_0$ be the upper left and right corner of $\canon{u}{v_0}$ (see Figure~\ref{fig:ConvexChain}). By definition, $\delta(u, w) \leq |u v_0| + \delta(v_0, w)$, and by the triangle inequality, $|u v_0| \leq \min\{|u a_0| + |a_0 v_0|, |u b_0| + |b_0 v_0|\}$. We assume without loss of generality that $v_0$ lies to the left of $u w$.

  Since $u w$ and $u v_0$ are visibility edges, by applying Lemma~\ref{lem:ConvexChain} to triangle $v_0 u w$,   a convex chain $v_0, ..., v_l = w$ of visibility edges connecting $v_0$ and $w$ exists (see Figure~\ref{fig:ConvexChain}). Note that, since $v_0$ is the closest visible vertex to $u$, every vertex along the convex chain lies above the horizontal line through $v_0$. 

  We now look at two consecutive vertices $v_{j-1}$ and $v_j$ along the convex chain. When $v_j \not \in C_0^{v_{j-1}}$, let $c$ and $d$ be the left and right corners of \canon{v_{j-1}}{v_j}. We distinguish four types of configurations: (i) $v_j \in C_i^{v_{j-1}}$ where $i > k$, or $i = k$ and $|c w| > |d w|$, (ii) $v_j \in C_i^{v_{j-1}}$ where $1 \leq i \leq k-1$, or $i = k$ and $|c w| \leq |d w|$, (iii) $v_j \in C_0^{v_{j-1}}$ and $v_j$ lies to the right of or has the same $x$-coordinate as $v_{j-1}$, and (iv) $v_j \in C_0^{v_{j-1}}$ and $v_j$ lies to the left of $v_{j-1}$. By convexity, the direction of $\overrightarrow{v_j v_{j+1}}$ is rotating counterclockwise for increasing $j$. Thus, these configurations occur in the order Type (i), Type (ii), Type (iii), Type (iv) along the convex chain from $v_0$ to $w$. We bound $\delta(v_{j-1}, v_j)$ as follows:

  \textbf{Type (i):} $v_j \in C_i^{v_{j-1}}$ where $i > k$, or $i = k$ and $|c w| > |d w|$. Since $v_j$ can see $v_{j-1}$ and \canon{v_j}{v_{j-1}} has smaller area than \canon{u}{w}, the induction hypothesis gives that $\delta(v_{j-1}, v_j)$ is at most $\max\{|v_{j-1} c| + \const \cdot |c v_j|, |v_{j-1} d| + \const \cdot |d v_j|\}$. 

  Let $a_j$ is the intersection of the horizontal line through $v_j$ and the left boundary of $C_0^{v_{j-1}}$. We aim to show that $\max\{|v_{j-1} c| + \const \cdot |c v_j|, |v_{j-1} d| + \const \cdot |d v_j|\} \leq |v_{j-1} a_j| + \const \cdot |a_j v_j|$. We use Lemma~\ref{lem:CalculationCase} to do this. However, since the precise application of this lemma depends on the family of $\theta$-graphs and determines the value of \const, this case is discussed in the spanning proofs of the three families. 

  \textbf{Type (ii):} $v_j \in C_i^{v_{j-1}}$ where $1 \leq i \leq k-1$, or $i = k$ and $|c w| \leq |d w|$. Since $v_j$ can see $v_{j-1}$ and \canon{v_j}{v_{j-1}} has smaller area than \canon{u}{w}, the induction hypothesis gives that $\delta(v_{j-1}, v_j)$ is at most $\max\{|v_{j-1} c| + \const \cdot |c v_j|, |v_{j-1} d| + \const \cdot |d v_j|\}$. 

  Let $a_j$ be the intersection of the left boundary of $C_0^{v_{j-1}}$ and the horizontal line through $v_j$. Since $v_j \in C_i^{v_{j-1}}$ where $1 \leq i \leq k-1$, or $i = k$ and $|c w| \leq |d w|$, we can apply Lemma~\ref{lem:ApplyFourPoints} in this case (where $v$, $w$, and $a$ from Lemma~\ref{lem:ApplyFourPoints} are $v_{j-1}$, $v_j$, and $a_j$) and we get that $\max\{|v_{j-1} c| + |c v_j|, |v_{j-1} d| + |d v_j|\} \leq |v_{j-1} a_j| + |a_j v_j|$ and $\max\{|c v_j|, |d v_j|\} \leq |a_j v_j|$. Since $\const \geq 1$, this implies that $\max\{|v_{j-1} c| + \const \cdot |c v_j|, |v_{j-1} d| + \const \cdot |d v_j|\} \leq |v_{j-1} a_j| + \const \cdot |a_j v_j|$. 

  \textbf{Type (iii):} If $v_j \in C_0^{v_{j-1}}$ and $v_j$ lies to the right of or has the same $x$-coordinate as $v_{j-1}$, let $a_j$ and $b_j$ be the left and right corner of \canon{v_{j-1}}{v_j}. Since $v_j$ can see $v_{j-1}$ and \canon{v_{j-1}}{v_j} has smaller area than \canon{u}{w}, we can apply the induction hypothesis. Thus, since $v_j$ lies to the right of or has the same $x$-coordinate as $v_{j-1}$, $\delta(v_{j-1}, v_j)$ is at most $|v_{j-1} a_j| + \const \cdot |a_j v_j|$.

  \textbf{Type (iv):} If $v_j \in C_0^{v_{j-1}}$ and $v_j$ lies to the left of $v_{j-1}$, let $a_j$ and $b_j$ be the left and right corner of \canon{v_{j-1}}{v_j}. Since $v_j$ can see $v_{j-1}$ and \canon{v_{j-1}}{v_j} has smaller area than \canon{u}{w}, we can apply the induction hypothesis. Thus, since $v_j$ lies to the left of $v_{j-1}$, $\delta(v_{j-1}, v_j)$ is at most $|v_{j-1} b_j| + \const \cdot |b_j v_j|$.

  Now that we have bounded the length of the inductive path for each type of configuration, we use these configurations to bound the total length of the path. We consider two cases: (a) $\angle a w u \leq \pi/2$, and (b) $\angle a w u > \pi/2$. 

  \textbf{Case (a):} We need to prove that $\delta(u, w) \leq \max\{|ua| + |aw|, |ub| + |bw|\} = |ua| + |aw|$. We first show that the convex chain cannot contain any Type (iv) configurations: for Type~(iv) configurations to occur, $v_j$ needs to lie to the left of $v_{j-1}$. However, by construction, $v_j$ lies on or to the right of the line through $v_{j-1}$ and $w$. Hence, since $\angle a w v_{j-1} < \angle a w u \leq \pi/2$, $v_j$ lies to the right of $v_{j-1}$. We can now bound $\delta(u, w)$ by using these bounds: 
  \begin{eqnarray*}
    \delta(u, w) &\leq& |u v_0| + \sum_{j=1}^l \delta(v_{j-1}, v_j) \\
		 &\leq& |u a_0| + |a_0 v_0| + \sum_{j=1}^l (|v_{j-1} a_j| + \const \cdot |a_j v_j|) \\
		 &\leq& |u a| + \const \cdot |a w|.
  \end{eqnarray*}

  \textbf{Case (b):} If $\angle a w u > \pi/2$, the convex chain can contain Type (iv) configurations. We need to prove that $\delta(u, w) \leq \max\{|ua| + |aw|, |ub| + |bw|\} = |ub| + |bw|$. Let \canon{v_{j'}}{v_{j'+1}} be the first Type (iv) configuration along the convex chain (if it has any), let $a'$ and $b'$ be the upper left and right corner of \canon{u}{v_{j'}}, and let $b''$ be the upper right corner of \canon{v_{j'}}{w}. We now bound $\delta(u, w)$ as follows (see Figure~\ref{fig:SpanningProofCase3-4k+2}): 
  \begin{eqnarray*}
    \delta(u, w) &\leq& |u v_0| + \sum_{j=1}^l \delta(v_{j-1}, v_j) \\
		 &\leq& |u a_0| + |a_0 v_0| + \sum_{j=1}^{j'} (|v_{j-1} a_j| + \const \cdot |a_j v_j|) + \sum_{j=j'+1}^l (|v_{j-1} b_j| + \const \cdot |b_j v_j|) \\
		 &\leq& |u a'| + \const \cdot |a' v_{j'}| + |v_{j'} b''| + \const \cdot |b'' w| \\
		 &\leq& |u b'| + \const \cdot |b' v_{j'}| + |v_{j'} b''| + \const \cdot |b'' w| \\
		 &=& |u b| + \const \cdot |b w|. 
  \end{eqnarray*}
  
  Note that it remains to prove Case~(i) for the three families of $\theta$-graphs, with their appropriate values of \const. Cases~(ii)-(iv), on the other hand, required only that $\const \geq 1$ and could therefore be handled for all three families at the same time. 
\end{proof}

\subsection[The Constrained \Graph{4}]{The Constrained $\boldsymbol{\theta_{(4 k + 4)}}$-Graph}
In this section we complete the proof of Theorem~\ref{theo:PathLengthGenericConstrained} for the constrained \graph{4}. 

\begin{theorem}
\label{theo:PathLength4k+4Constrained}
  Let $u$ and $w$ be two vertices in the plane such that $u$ can see $w$. Let $m$ be the midpoint of the side of \canon{u}{w} opposite $u$ and let $\alpha$ be the unsigned angle between $u w$ and $u m$. There exists a path connecting $u$ and $w$ in the constrained \graph{4} of length at most 
  \[\left( \frac{\cos \alpha}{\cos \left(\frac{\theta}{2}\right)} + \frac{\cos \alpha \cdot \tan \left(\frac{\theta}{2}\right) + \sin \alpha}{\cos \left(\frac{\theta}{2}\right) - \sin \left(\frac{\theta}{2}\right)} \right) \cdot |u w|.\] 
\end{theorem}
\begin{proof} 
  We apply Theorem~\ref{theo:PathLengthGenericConstrained} using $\const = 1 / \left(\cos (\theta/2) - \sin (\theta/2) \right)$. The assumptions made in Theorem~\ref{theo:PathLengthGenericConstrained} still apply. Recall that $c$ and $d$ are the left and right corners of \canon{v_{j-1}}{v_j}, opposite to $v$, and $a_j$ is the intersection of the horizontal line through $v_j$ and the left boundary of $C_0^{v_{j-1}}$. It remains to show that for the Type (i) configurations, we have that $\max\{|v_{j-1} c| + \const \cdot |c v_j|, |v_{j-1} d| + \const \cdot |d v_j|\} \leq |v_{j-1} a_j| + \const \cdot |a_j v_j|$. Let $\beta$ be $\angle a_j v_j v_{j-1}$ and let $\gamma$ be the angle between $v_j v_{j-1}$ and the bisector of \canon{v_{j-1}}{v_j}. 

  \begin{figure}[ht]
    \begin{center}
      \includegraphics{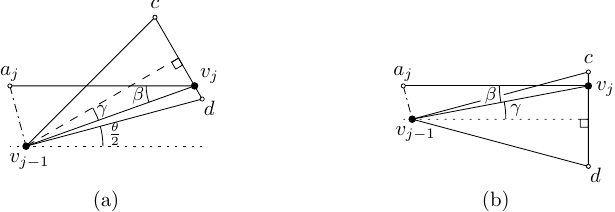}
    \end{center}
    \caption{The remaining cases of the induction step for the \graph{4}: (a) $v_j \in C_k^{v_{j-1}}$ and $|c w| > |d w|$, and (b) $v_j \in C_{k+1}^{v_{j-1}}$}
    \label{fig:SpanningProof4k+4Constrained}
  \end{figure}

  We distinguish two cases: (a) $v_j \in C_k^{v_{j-1}}$ and $|c w| > |d w|$, and (b) $v_j \in C_{k+1}^{v_{j-1}}$. 

  \textbf{Case (a):} When $v_j \in C_k^{v_{j-1}}$ and $|c w| > |d w|$, the induction hypothesis for \canon{v_{j-1}}{v_j} gives $\delta(v_{j-1}, v_j) \leq |v_{j-1} c| + \const \cdot |c v_j|$ (see Figure~\ref{fig:SpanningProof4k+4Constrained}a). We note that $\gamma = \theta - \beta$. Hence Lemma~\ref{lem:CalculationCase} gives that the inequality holds when $\const \geq (\cos (\theta - \beta) - \sin \beta) / (\cos (\theta/2 - \beta) - \sin (3\theta/2 - \beta))$. As this function is decreasing in $\beta$ for $\theta/2 \leq \beta \leq \theta$, it is maximized when $\beta$ equals $\theta/2$. Hence $\const$ needs to be at least $(\cos (\theta/2) - \sin (\theta/2)) / (1 - \sin \theta)$, which can be rewritten to $1 / (\cos (\theta/2) - \sin (\theta/2))$. 

  \textbf{Case (b):} When $v_j \in C_{k+1}^{v_{j-1}}$, $v_j$ lies above the bisector of \canon{v_{j-1}}{v_j} and the induction hypothesis for \canon{v_{j-1}}{v_j} gives $\delta(v_{j-1}, v_j) \leq |v_j d| + \const \cdot |d v_{j-1}|$ (see Figure~\ref{fig:SpanningProof4k+4Constrained}b). We note that $\gamma = \beta$. Hence Lemma~\ref{lem:CalculationCase} gives that the inequality holds when $\const \geq (\cos \beta - \sin \beta) / (\cos (\theta/2 - \beta) - \sin (\theta/2 + \beta))$, which is equal to $1 / (\cos (\theta/2) - \sin (\theta/2))$.
\end{proof}

Since $\cos \alpha / \cos (\theta/2) + (\cos \alpha \tan (\theta/2) + \sin \alpha) / (\cos (\theta/2) - \sin (\theta/2))$ is increasing in $\alpha$, for $\alpha \in [0, \theta/2]$ and fixed $\theta \in [0, \pi/4]$, it is maximized when $\alpha = \theta/2$, and we obtain the following corollary: 

\begin{corollary}
  \label{cor:SpanningRatio4k+4Constrained}
  The constrained \graph{4} is a $\left( 1 + \frac{2 \cdot \sin \left( \frac{\theta}{2} \right)}{\cos \left( \frac{\theta}{2} \right) - \sin \left( \frac{\theta}{2} \right)} \right)$-spanner of $\Vis(P,S)$. 
\end{corollary}

\subsection[The Constrained \Graph{3} and \Graph{5}]{The Constrained $\boldsymbol{\theta_{(4 k + 3)}}$-Graph and $\boldsymbol{\theta_{(4 k + 5)}}$-Graph}
In this section we complete the proof of Theorem~\ref{theo:PathLengthGenericConstrained} for the constrained \graph{3} and \graph{5}. 

\begin{theorem}
  \label{theo:PathLength4k+3Constrained}
  Let $u$ and $w$ be two vertices in the plane such that $u$ can see $w$. Let $m$ be the midpoint of the side of \canon{u}{w} opposite $u$ and let $\alpha$ be the unsigned angle between $u w$ and $u m$. There exists a path connecting $u$ and $w$ in the constrained \graph{3} of length at most 
  \[\left( \frac{\cos \alpha}{\cos \left(\frac{\theta}{2}\right)} + \frac{\left( \cos \alpha \cdot \tan \left(\frac{\theta}{2}\right) + \sin \alpha \right) \cdot \cos \left(\frac{\theta}{4}\right)}{\cos \left(\frac{\theta}{2}\right) - \sin \left(\frac{3\theta}{4}\right)} \right) \cdot |u w|.\] 
\end{theorem}
\begin{proof}
  We apply Theorem~\ref{theo:PathLengthGenericConstrained} using $\const = \cos (\theta/4) / (\cos (\theta/2) - \sin (3\theta/4))$. The assumptions made in Theorem~\ref{theo:PathLengthGenericConstrained} still apply. Recall that $c$ and $d$ are the left and right corners of \canon{v_{j-1}}{v_j}, opposite to $v$, and $a_j$ is the intersection of the horizontal line through $v_j$ and the left boundary of $C_0^{v_{j-1}}$. It remains to show that for the Type (i) configurations, we have that $\max\{|v_{j-1} c| + \const \cdot |c v_j|, |v_{j-1} d| + \const \cdot |d v_j|\} \leq |v_{j-1} a_j| + \const \cdot |a_j v_j|$. Let $\beta$ be $\angle a_j v_j v_{j-1}$ and let $\gamma$ be the angle between $v_j v_{j-1}$ and the bisector of \canon{v_{j-1}}{v_j}. 

  \begin{figure}[ht]
    \begin{center}
      \includegraphics{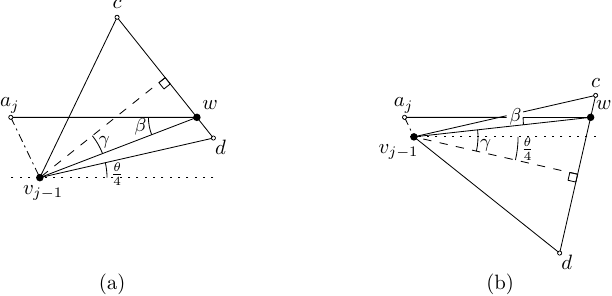}
    \end{center}
    \caption{The remaining cases of the induction step for the \graph{3}: (a) $v_j \in C_k^{v_{j-1}}$ and $|c w| > |d w|$, and (b) $v_j \in C_{k+1}^{v_{j-1}}$}
    \label{fig:SpanningProof4k+3Constrained}
  \end{figure}

  We distinguish two cases: (a) $v_j \in C_k^{v_{j-1}}$ and $|c w| > |d w|$, and (b) $v_j \in C_{k+1}^{v_{j-1}}$. 

  \textbf{Case (a):} When $v_j \in C_k^{v_{j-1}}$ and $|c w| > |d w|$, the induction hypothesis for \canon{v_{j-1}}{v_j} gives $\delta(v_{j-1}, v_j) \leq |v_{j-1} c| + \const \cdot |c v_j|$ (see Figure~\ref{fig:SpanningProof4k+3Constrained}a). We note that $\gamma = 3\theta/4 - \beta$. Hence Lemma~\ref{lem:CalculationCase} gives that the inequality holds when $\const \geq (\cos (3\theta/4 - \beta) - \sin \beta) / (\cos (\theta/2 - \beta) - \sin (5\theta/4 - \beta))$. As this function is decreasing in $\beta$ for $\theta/4 \leq \beta \leq 3\theta/4$, it is maximized when $\beta$ equals $\theta/4$. Hence $\const$ needs to be at least $(\cos (\theta/2) - \sin (\theta/4)) / (\cos (\theta/4) - \sin \theta)$, which is equal to $\cos (\theta/4) / (\cos (\theta/2) - \sin (3\theta/4))$. 

  \textbf{Case (b):} When $v_j \in C_{k+1}^{v_{j-1}}$, $v_j$ lies above the bisector of \canon{v_{j-1}}{v_j} and the induction hypothesis for \canon{v_{j-1}}{v_j} gives $\delta(v_{j-1}, v_j) \leq |v_j d| + \const \cdot |d v_{j-1}|$ (see Figure~\ref{fig:SpanningProof4k+3Constrained}b). We note that $\gamma = \theta/4 + \beta$. Hence Lemma~\ref{lem:CalculationCase} gives that the inequality holds when $\const \geq (\cos (\theta/4 + \beta) - \sin \beta) / (\cos (\theta/2 - \beta) - \sin (3\theta/4 + \beta))$, which is equal to $\cos (\theta/4) / (\cos (\theta/2) - \sin (3\theta/4))$. 
\end{proof}

\begin{theorem}
  \label{theo:PathLength4k+5Constrained}
  Let $u$ and $w$ be two vertices in the plane such that $u$ can see $w$. Let $m$ be the midpoint of the side of \canon{u}{w} opposite $u$ and let $\alpha$ be the unsigned angle between $u w$ and $u m$. There exists a path connecting $u$ and $w$ in the constrained \graph{5} of length at most 
  \[\left( \frac{\cos \alpha}{\cos \left(\frac{\theta}{2}\right)} + \frac{\left( \cos \alpha \cdot \tan \left(\frac{\theta}{2}\right) + \sin \alpha \right) \cdot \cos \left(\frac{\theta}{4}\right)}{\cos \left(\frac{\theta}{2}\right) - \sin \left(\frac{3\theta}{4}\right)} \right) \cdot |u w|.\] 
\end{theorem}
\begin{proof}
  We apply Theorem~\ref{theo:PathLengthGenericConstrained} using $\const = \cos (\theta/4) / (\cos (\theta/2) - \sin (3\theta/4))$. The assumptions made in Theorem~\ref{theo:PathLengthGenericConstrained} still apply. Recall that $c$ and $d$ are the left and right corners of \canon{v_{j-1}}{v_j}, opposite to $v$, and $a_j$ is the intersection of the horizontal line through $v_j$ and the left boundary of $C_0^{v_{j-1}}$. It remains to show that for the Type (i) configurations, we have that $\max\{|v_{j-1} c| + \const \cdot |c v_j|, |v_{j-1} d| + \const \cdot |d v_j|\} \leq |v_{j-1} a_j| + \const \cdot |a_j v_j|$. Let $\beta$ be $\angle a_j v_j v_{j-1}$ and let $\gamma$ be the angle between $v_j v_{j-1}$ and the bisector of \canon{v_{j-1}}{v_j}. 

  \begin{figure}[ht]
    \begin{center}
      \includegraphics{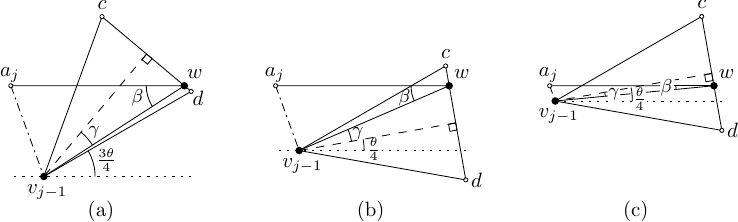}
    \end{center}
    \caption{The remaining cases of the induction step for the \graph{5}: (a) $w$ lies in $C_k^v$ and $|c w| > |d w|$, (b) $w$ lies in $C_{k+1}^v$ and $|c w| < |d w|$, and (c) $w$ lies in $C_{k+1}^v$ and $|c w| \geq |d w|$}
    \label{fig:SpanningProof4k+5Constrained}
  \end{figure}

  We distinguish two cases: (a) $v_j \in C_k^{v_{j-1}}$ and $|c w| > |d w|$, and (b) $v_j \in C_{k+1}^{v_{j-1}}$. 

  \textbf{Case (a):} When $v_j \in C_k^{v_{j-1}}$ and $|c w| > |d w|$, the induction hypothesis for \canon{v_{j-1}}{v_j} gives $\delta(v_{j-1}, v_j) \leq |v_{j-1} c| + \const \cdot |c v_j|$ (see Figure~\ref{fig:SpanningProof4k+5Constrained}a). We note that $\gamma = 5\theta/4 - \beta$. Hence Lemma~\ref{lem:CalculationCase} gives that the inequality holds when $\const \geq (\cos (5\theta/4 - \beta) - \sin \beta) / (\cos (\theta/2 - \beta) - \sin (7\theta/4 - \beta))$. As this function is decreasing in $\beta$ for $3\theta/4 \leq \beta \leq 5\theta/4$, it is maximized when $\beta$ equals $3\theta/4$. Hence $\const$ needs to be at least $(\cos (\theta/2) - \sin (3\theta/4)) / (\cos (\theta/4) - \sin \theta)$, which is less than $\cos (\theta/4) / (\cos (\theta/2) - \sin (3\theta/4))$. 

  \textbf{Case (b):} When $v_j \in C_{k+1}^{v_{j-1}}$, the induction hypothesis for \canon{v}{w} gives $\delta(v_{j-1}, v_j) \leq \max\{|v_{j-1} c| + \const \cdot |c v_j|, |v_{j-1} d| + \const \cdot |d v_j|\}$. If $|c w| < |d w|$ (see Figure~\ref{fig:SpanningProof4k+5Constrained}b), we note that $\gamma = \beta - \theta/4$. Hence Lemma~\ref{lem:CalculationCase} gives that the inequality holds when $\const \geq (\cos (\beta - \theta/4) - \sin \beta) / (\cos (\theta/2 - \beta) - \sin (\theta/4 + \beta))$, which is equal to $\cos (\theta/4) / (\cos (\theta/2) - \sin (3\theta/4))$.

  If $|d w| < |c w|$ (see Figure~\ref{fig:SpanningProof4k+5Constrained}c), we note that $\gamma = \theta/4 - \beta$. Hence Lemma~\ref{lem:CalculationCase} gives that the inequality holds when $\const \geq (\cos (\theta/4 - \beta) - \sin \beta) / (\cos (\theta/2 - \beta) - \sin (3\theta/4 - \beta))$. As this function is decreasing in $\beta$ for $0 \leq \beta \leq \theta/4$, it is maximized when $\beta$ equals $0$. Hence $\const$ needs to be at least $\cos (\theta/4) / (\cos (\theta/2) - \sin (3\theta/4))$.
\end{proof}

When looking at two vertices $u$ and $w$ in the constrained \graph{3} and \graph{5}, we notice that when the angle between $u w$ and the bisector of \canon{u}{w} is $\alpha$, the angle between $w u$ and the bisector of \canon{w}{u} is $\theta/2 - \alpha$. Hence the worst case spanning ratio becomes the minimum of the spanning ratio when looking at \canon{u}{w} and the spanning ratio when looking at \canon{w}{u}. 

\begin{theorem}
  \label{theo:SpanningRatio4k+3,5Constrained}
  The constrained \graph{3} and \graph{5} are $\frac{\cos \left(\frac{\theta}{4}\right)}{\cos \left(\frac{\theta}{2}\right) - \sin \left(\frac{3\theta}{4}\right)}$-spanners of $\Vis(P,S)$. 
\end{theorem}
\begin{proof}
  The spanning ratio of the constrained \graph{3} and \graph{5} is at most: 
  \[ \min \left\{
  \begin{array}{l}
    \frac{\cos \alpha}{\cos \left(\frac{\theta}{2}\right)} + \frac{\left( \cos \alpha \cdot \tan \left(\frac{\theta}{2}\right) + \sin \alpha \right) \cdot \cos \left(\frac{\theta}{4}\right)}{\cos \left(\frac{\theta}{2}\right) - \sin \left(\frac{3\theta}{4}\right)}, \\
    \frac{\cos \left(\frac{\theta}{2} - \alpha\right)}{\cos \left(\frac{\theta}{2}\right)} + \frac{\left( \cos \left(\frac{\theta}{2} - \alpha\right) \cdot \tan \left(\frac{\theta}{2}\right) + \sin \left(\frac{\theta}{2} - \alpha\right) \right) \cdot \cos \left(\frac{\theta}{4}\right)}{\cos \left(\frac{\theta}{2}\right) - \sin \left(\frac{3\theta}{4}\right)}
  \end{array}
  \right\}
  \]

  Since $\cos \alpha / \cos \left( \frac{\theta}{2} \right) + \const \cdot \left( \cos \alpha \cdot \tan \left( \frac{\theta}{2} \right) + \sin \alpha \right)$ is increasing in $\alpha$, for $\alpha \in [0, \theta/2]$ and fixed $\theta \in [0, 2\pi/7]$, the minimum of these two functions is maximized when the two functions are equal, i.e. when $\alpha = \theta/4$. Thus the constrained \graph{3} and \graph{5} have spanning ratio at most: \[\frac{\cos \left(\frac{\theta}{4}\right)}{\cos \left(\frac{\theta}{2}\right)} + \frac{\left( \cos \left(\frac{\theta}{4}\right) \cdot \tan \left(\frac{\theta}{2}\right) + \sin \left(\frac{\theta}{4}\right) \right) \cdot \cos \left(\frac{\theta}{4}\right)}{\cos \left(\frac{\theta}{2}\right) - \sin \left(\frac{3\theta}{4}\right)} = \frac{\cos \left(\frac{\theta}{4}\right) \cdot \cos \left(\frac{\theta}{2}\right)}{\cos \left(\frac{\theta}{2}\right) \cdot \left( \cos \left(\frac{\theta}{2}\right) - \sin \left(\frac{3\theta}{4}\right) \right)}.\] 

\end{proof}

\section{Constrained Yao-Graphs}
In this section, we prove that constrained Yao-graphs with at least 7 cones are spanners of the visibility graph. 

\begin{theorem}
  \label{theo:SpanningRatio}
  The \ygraph ($m \geq 7$) is a $1 / \left( 1 - 2 \sin \left( \frac{\theta}{2} \right) \right)$-spanner of $\Vis(P,S)$. 
\end{theorem}
\begin{proof}
  Let $u$ and $w$ be two vertices that can see each other. We show that there exists a path connecting $u$ and $w$ in the \ygraph ($m \geq 7$) of length at most $t \cdot |u w|$ for $t = 1/(1 - 2 \sin (\theta/2))$, by induction on the rank of the distance between every pair of vertices $u$ and $w$ that can see each other. For ease of exposition, we assume without loss of generality that $w \in C_0^u$. 

  \textbf{Base case:} Vertices $u$ and $w$ are a closest visible pair. Since the closest visible pair need not be unique, we proceed to show that the subcone of $C_0^u$ that contains $w$ does not contain any vertices visible to $u$ at distance at most $|u w|$: If there were such a vertex $x$, since $u x$ and $x w$ are visibility edges that lie in the same subcone, by Lemma~\ref{lem:ConvexChain} there exists a convex chain of visibility edges connecting $x$ to $w$. Since we have at least 7 cones, the vertex adjacent to $w$ along this chain is strictly closer to $w$ than $u$, contradicting that $|u w|$ is a closest visible pair. Hence, since $w$ is the closest visible vertex, $u w$ is an edge in the \ygraph and thus there exists a path between $u$ and $w$ of length $|u w| < t \cdot |u w|$. 

  \textbf{Induction step:} We assume that the induction hypothesis holds for all pairs of vertices that can see each other and whose distance is less than $|u w|$. 

  If $u w$ is an edge in the \ygraph, the induction hypothesis follows by the same argument as in the base case. If there is no edge between $u$ and $w$, let $v$ be the closest visible vertex to $u$ in the subcone of $u$ that contains $w$, and let $x$ be the point along $u w$ such that $|u v| = |u x|$ (see Figure~\ref{fig:ConvexChainYao}). Since $x$ lies on $u w$, both $u x$ and $x w$ are visibility edges. 

  \begin{figure}[ht]
    \begin{center}
      \includegraphics{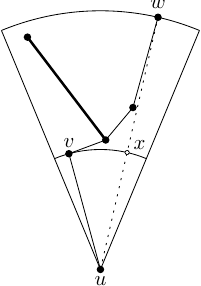}
    \end{center}
    \caption{A convex chain from $v$ to $w$}
    \label{fig:ConvexChainYao}
  \end{figure}

  Next, we show that $v x$ is also a visibility edge: If $v x$ is not a visibility edge, that implies that it crosses some constraint. Since $u v$ and $u x$ are visibility edges, this constraint cannot cross them. Therefore, one endpoint of the constraint is contained in triangle $u v x$. Let $y$ be this endpoint. Since $v$ and $w$ lie in the same subcone of $u$, $u$ is not the endpoint of a constraint intersecting the interior of $u v x$. Hence, we can apply Lemma~\ref{lem:ConvexChain} and obtain a convex chain of visibility edges from $v$ and $x$ and the polygon defined by $u v$, $u x$, and the convex chain is empty and does not contain any constraints. This implies that $u$ can see every vertex along the convex chain, each of which is closer to it than $v$, contradicting that $v$ was the closest visible vertex to $u$. 

  Since $v x$ and $x w$ are visibility edges, we can apply Lemma~\ref{lem:ConvexChain} to triangle $v x w$ and we obtain a convex chain of visibility edges $v = p_0, ..., p_j = w$ connecting $v$ and $w$ (see Figure~\ref{fig:ConvexChainYao}). Since we have at least 7 cones, the distance between any two consecutive vertices is strictly less than $|u w|$. Hence, since every pair of consecutive vertices along this convex chain can see each other, we can apply induction on each of them. Therefore, there exists a path from $u$ to $w$ via $v$ of length at most \[|u v| + t \cdot \sum_{i=0}^{j-1} |p_i p_{i+1}|.\]

  Since the chain between $v$ and $w$ is contained in triangle $v x w$ and the chain is convex, it follows that the total length of the chain is at most $|v x| + |x w|$. Thus, we can upper bound the length of the path by \[|u v| + t \cdot \left( |v x| + |x w| \right).\]

  Since $|u v| = |u x|$, triangle $u v x$ is an isosceles triangle and we can express $|v x|$ as $2 \sin \left( \angle v u x/2 \right) \cdot |u v|$. Since this function is increasing in $\angle v u x$, for $\angle v u x \in [0, 2\pi/7]$ and $\angle v u x \in [0, \theta]$, it follows that $|v x| \leq 2 \sin \left( \theta/2 \right) \cdot |u v|$. Next, we look at $|x w|$: Since $x$ lies on $u w$ and $|u v| = |u x|$, it follows that $|x w| = |u w| - |u x| = |u w| - |u v|$. Hence, the path between $u$ and $w$ has length at most
  \begin{align*}
    & ~~|u v| + t \cdot \left( |v x| + |x w| \right) \\
    \leq & ~~|u v| + t \cdot \left( 2 \sin \left( \frac{\theta}{2} \right) \cdot |u v| + |u w| - |u v| \right) \\
    = & ~~t \cdot |u w| + \left( 1 + 2 \sin \left( \frac{\theta}{2} \right) \cdot t - t \right) \cdot |u v|.
  \end{align*}

  Hence, for the length of the path to be at most $t \cdot |u w|$, we need that \[1 + 2 \sin \left( \frac{\theta}{2} \right) \cdot t - t \leq 0,\] which can be rewritten to \[t \geq \frac{1}{1 - 2 \sin \left( \frac{\theta}{2} \right)},\] completing the proof. 
\end{proof}

For odd values of $m$, the spanning ratio can be decreased a bit: Let $C_i^u$ be the cone of $u$ that contains $w$ and let $C_j^w$ be the cone of $w$ that contains $u$. When we look at two vertices $u$ and $w$ in the \ygraph, we notice that when the angle between $u w$ and the bisector of $C_i^u$ is $\alpha$, the angle between $w u$ and the bisector of $C_j^w$ is $\theta/2 - \alpha$ (see Figure~\ref{fig:OddYao}). Hence, when bounding the worst case spanning ratio of \ygraph{s} with an odd number of cones, we can assume without loss of generality that the angle between the bisector of the cone and $u w$ is at most $\theta/4$. 

\begin{figure}[ht]
  \begin{center}
    \includegraphics{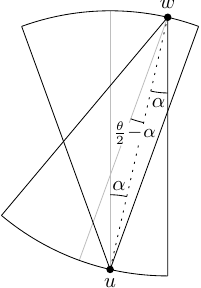}
  \end{center}
  \caption{The angle between $u w$ and the bisector of $C_i^u$ is $\alpha$ and the angle between $w u$ and the bisector of $C_j^w$ is $\theta/2 - \alpha$}
  \label{fig:OddYao}
\end{figure}

Because of this property, we can also extend the theorem to the $Y_5$-graph. This is surprising, since the proof of Theorem~\ref{theo:SpanningRatio} cannot be applied easily to Yao-graphs with fewer than 7 cones. 

\begin{theorem}
  \label{theo:SpanningRatioOdd}
  For odd values of $m \geq 5$, the \ygraph is a $1 / \left( 1 - 2 \sin \left( \frac{3\theta}{8} \right) \right)$-spanner of $\Vis(P,S)$. 
\end{theorem}
\begin{proof}
  Let $u$ and $w$ be two vertices that can see each other. We show that there exists a path connecting $u$ and $w$ in the \ygraph ($m \geq 5$) of length at most $t \cdot |u w|$ for $t = 1/(1 - 2 \sin (3\theta/8))$, by induction on the rank of the distance between every pair of vertices $u$ and $w$ that can see each other. For ease of exposition, we assume without loss of generality that $w \in C_0^u$. We also assume without loss of generality that the angle between the bisector of $C_0^u$ and $u w$ is at most $\theta/4$. 

  \textbf{Base case:} Vertices $u$ and $w$ are a closest visible pair. Using the same argument as in Theorem~\ref{theo:SpanningRatio}, it follows that $u w$ is an edge of the \ygraph and thus there exists a path between $u$ and $w$ of length $|u w| < t \cdot |u w|$. 

  \textbf{Induction step:} We assume that the induction hypothesis holds for all pairs of vertices that can see each other and whose distance is less than $|u w|$. 

  If $u w$ is an edge in the \ygraph, the induction hypothesis follows by the same argument as in the base case. If there is no edge between $u$ and $w$, let $v$ be the closest visible vertex to $u$ in the subcone of $u$ that contains $w$, and let $x$ be the point along $u w$ such that $|u v| = |u x|$ (see Figure~\ref{fig:ConvexChainOddYao}). Since $x$ lies on $u w$, both $u x$ and $x w$ are visibility edges.  

  \begin{figure}[ht]
    \begin{center}
      \includegraphics{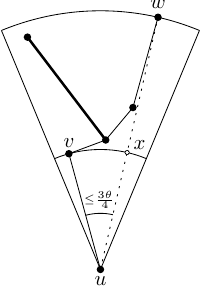}
    \end{center}
    \caption{A convex chain from $v$ to $w$}
    \label{fig:ConvexChainOddYao}
  \end{figure}

  Using the same argument as in Theorem~\ref{theo:SpanningRatio}, it follows that $v x$ is also a visibility edge. Hence, we can apply Lemma~\ref{lem:ConvexChain} to triangle $v x w$ and we obtain a convex chain of visibility edges $v = p_0, ..., p_j = w$ connecting $v$ and $w$ (see Figure~\ref{fig:ConvexChainOddYao}). Since we have at least 5 cones and the angle between the bisector of $C_0^u$ and $u w$ is at most $\theta/4$, the distance between any two consecutive vertices is strictly less than $|u w|$. Hence, since every pair of consecutive vertices along this convex chain can see each other, we can apply induction on each of them. Therefore, there exists a path from $u$ to $w$ via $v$ of length at most \[|u v| + t \cdot \sum_{i=0}^{j-1} |p_i p_{i+1}|.\] Analogous to Theorem~\ref{theo:SpanningRatio}, this expression can be upper bounded by $|u v| + t \cdot \left( |v x| + |x w| \right)$.

  Since $|u v| = |u x|$, triangle $u v x$ is an isosceles triangle and we can express $|v x|$ as $2 \sin \left( \angle v u x/2 \right) \cdot |u v|$. Since this function is increasing in $\angle v u x$, for $\angle v u x \in [0, 3\theta/4]$  and fixed $\theta \in [0, 2\pi/5]$, it follows that $|v x| \leq 2 \sin \left( 3\theta/8 \right) \cdot |u v|$. Analogous to Theorem~\ref{theo:SpanningRatio}, it holds that $|x w| = |u w| - |u v|$. Hence, the path between $u$ and $w$ has length at most
  \begin{align*}
    & ~~|u v| + t \cdot \left( |v x| + |x w| \right) \\
    \leq & ~~|u v| + t \cdot \left( 2 \sin \left( \frac{3\theta}{8} \right) \cdot |u v| + |u w| - |u v| \right) \\
    = & ~~t \cdot |u w| + \left( 1 + 2 \sin \left( \frac{3\theta}{8} \right) \cdot t - t \right) \cdot |u v|.
  \end{align*}

  Hence, for the length of the path to be at most $t \cdot |u w|$, we need that \[1 + 2 \sin \left( \frac{3\theta}{8} \right) \cdot t - t \leq 0,\] which can be rewritten to \[t \geq \frac{1}{1 - 2 \sin \left( \frac{3\theta}{8} \right)},\] completing the proof. 
\end{proof}

\section{Conclusion}
We showed that the constrained \graph{2} has a tight spanning ratio of $1 + 2 \sin(\theta/2)$. This is the first time tight spanning ratios have been found for a large family of constrained $\theta$-graphs. Previously, the only constrained $\theta$-graph for which tight bounds were known was the constrained $\theta_6$-graph. We also gave improved upper bounds on the spanning ratio of the constrained \graph{3}, the constrained \graph{4}, and the constrained \graph{5}. 

There remain a number of open problems, such as finding tight spanning ratios for the constrained \graph{3}, the constrained \graph{4}, and the constrained \graph{5}. Another set of open problems concerns constrained $\theta$-graphs with few cones. In the unconstrained setting, it is known that the $\theta_4$-graph and the $\theta_5$-graph are spanners, but this question remains unanswered in the constrained setting. 

We also looked at constrained Yao-graphs and showed that constrained Yao-graphs with 5 or at least 7 cones are spanners of the visibility graph. Furthermore, the upper bounds on the spanning ratio we obtained match those of the unconstrained Yao-graphs. However, since these bounds are not known to be tight, this raises a number of new questions, the obvious one being whether we can reduce the upper bounds or find matching lower bound constructions. 

Another set of open problems involves constrained Yao-graphs with 4 or 6 cones. In the unconstrained setting, it is known that the $Y_m$-graph is a spanner if and only if $m \geq 4$. Since the proof presented in this paper can be applied only to Yao-graphs with 5 or at least 7 cones, it remains unknown whether this is also true in the constrained setting. 

Finally, though we have upper bounds on the spanning ratio of $\theta$-graphs and Yao-graphs in the constrained setting, we do not have a local competitive routing algorithm to actually route messages between any two visible vertices. The main difficulty stems from the inductive steps along the convex chain, since these steps make it unclear where the routing algorithm should forward the message to. In particular, we cannot assume that there exists an edge in the subcone that contains the destination, since visibility may be blocked by a constraint. Hence, routing remains a major open problem in this area.

\bibliographystyle{plain}
\bibliography{references}
\end{document}